\documentclass[11pt,reqno]{amsart}

\usepackage{amscd,amssymb,amsmath,amsthm}
\usepackage[arrow,matrix]{xy}
\usepackage{graphicx}
\usepackage{color}
\usepackage{cite}
\newtheorem{thm}{Theorem}

\newtheorem{lemma}{Lemma}
\newtheorem{pro}{Proposition}

\newtheorem{cor}{Corollary}

\numberwithin{equation}{section} 

\begin{document}
\title[The $p$-adic Solid-on-Solid model]
{On $p$-adic Solid-on-Solid model on a Cayley tree}

\author{O. N. Khakimov}

\address{O.N. Khakimov\\ Institute of mathematics,
29, Do'rmon Yo'li str., 100125, Tashkent, Uzbekistan.}
\email {hakimovo@mail.ru}

\maketitle

\begin{abstract} We consider a nearest-neighbor $p$-adic Solid-on-Solid
(SOS) model
 with $m+1$ spin values and coupling constant $J\in\mathbb Q_p$ on a
Cayley tree. It is found conditions under which a phase
transition does not occur for this model. It is shown that under condition
$p\mid m+1$ for some $J$ a phase transition occurs. Moreover, we give
criterion of boundedness of
$p$-adic Gibbs measures for the $m+1$-state SOS model.
\end{abstract}
\maketitle


\section{introduction}

The $p$-adic numbers were first introduced by the German
mathematician K.Hensel. They form an integral part of number
theory, algebraic geometry,
representation theory and other branches of modern mathematics.
However, numerous applications of these numbers to
theoretical physics have been proposed \cite{14}
to quantum mechanics and to $p$-adic valued physical
observable \cite{k91}. A number of $p$-adic models in physics
cannot be described using ordinary probability theory based on the
Kolmogorov axioms.

In \cite{33} a theory of stochastic processes with values in
$p$-adic and more general non-Archimedean fields was developed,
having probability distributions with non-Archimedean values.

One of the basic branches of mathematics lying at the base of the
theory of statistical mechanics is the theory of probability and
stochastic processes. Since the theories of probability and
stochastic processes in a non-Archimedean setting have been
introduced, it is natural to study problems of statistical
mechanics in the context of the $p$-adic theory of probability.

We note that $p$-adic Gibbs measures were studied for several
$p$-adic
models of statistical mechanics \cite{GRR,16,TMP2},
\cite{MRasos,MNGM,rh,bookroz}.

The SOS model can be considered as a generalization of the Ising model
(which arises when $m=1$). It is known that a phase transition does not occur
for the $p$-adic Ising model on a Cayley tree. We prove that there is no phase transition
for the $m+1$-state $p$-adic SOS model if $p\nmid m+1$. Moreover, we show that if
$p\mid m+1$ then a phase transition may occur for this model.

The organization of this paper as follows. Section 2 is
a mathematically preliminary. In section 3 we give a construction
of $p$-adic Gibbs measures for the $m+1$-state $p$-adic SOS model
on a Cayley tree of order $k\geq1$. Moreover, we study
a problem of boundedness of such measures. In section 4 we study
the set of all translation-invariant $p$-adic Gibbs measures for
$3$-state $p$-adic SOS model. In section 5 we give a criterion of uniqueness
of the $p$-adic Gibbs measure.

\section{definitions and preliminary results}

\subsection{\bf $p$-adic numbers and measures.}

Let $\mathbb Q$ be the field of rational numbers.
For a fixed prime number $p$, every rational number $x\ne0$ can
be represented in the form $x = p^r{n\over m}$, where $r,n\in\mathbb Z$,
$m$ is a positive integer, and $n$ and $m$ are relatively
prime with $p$: $(p, n) = 1$, $(p, m) = 1$. The $p$-adic norm of
$x$ is given by
$$
|x|_p= \left\{\begin{array}{ll}
p^{-r},& \text{ if }  x\neq 0,\\
0,& \text{ if } x=0.
 \end{array}\right.
$$

This norm is non-Archimedean  and satisfies the so called strong
triangle inequality
$$|x+y|_p\leq \max\{|x|_p,|y|_p\}.$$

From this property immediately follow the following facts ({\it non-Archimedean norm's property}):

1) if  $|x|_p\neq |y|_p$, then $|x-y|_p=\max\{|x|_p,|y|_p\}$;

2) if  $|x|_p=|y|_p$, then  $|x-y|_p\leq |x|_p$;

The completion of $\mathbb Q$ with respect to the $p$-adic norm defines
the $p$-adic field $\mathbb Q_p$ (see \cite{29}).

The completion of the field of rational numbers $\mathbb Q$ is either
the field of real numbers $\mathbb R$ or one of the fields of
$p$-adic numbers $\mathbb Q_p$ (Ostrowski's theorem).

Any $p$-adic number $x\ne 0$ can be uniquely represented
in the canonical form
\begin{equation}\label{ek}
x = p^{\gamma(x)}(x_0+x_1p+x_2p^2+\dots),
\end{equation}
where $\gamma=\gamma(x)\in \mathbb Z$ and the integers $x_j$ satisfy: $x_0 > 0$,
$0\leq x_j \leq p - 1$ (see \cite{29,sc,48}). In this case $|x|_p =
p^{-\gamma(x)}$.
\begin{thm}\label{tx2}\cite{48}
The equation
$x^2 = a$, $0\ne a =p^{\gamma(a)}(a_0 + a_1p + ...), 0\leq a_j
\leq p - 1$, $a_0 > 0$ has a solution $x\in \mathbb Q_p$ iff
hold true the following:\\
$(a)$ $\gamma(a)$ is even;\\
$(b)$ $y^2\equiv a_0(\operatorname{mod} p)$ is solvable for $p\ne 2$;
the equality $a_1=a_2=0$ holds if $p=2$.
\end{thm}

We respectively denote the sets of all {\it $p$-adic integers}
and {\it units} of $\mathbb Q_p$ by
$$
\mathbb Z_p=\left\{x\in\mathbb Q_p: |x|_p\leq1\right\},\quad
\mathbb Z^*_p=\left\{x\in\mathbb Q_p: |x|_p=1\right\}.
$$
\begin{lemma}\label{lem1}(Hensel's lemma, \cite{48}). Let $f(x)$ be a polynomial
whose the coefficients are $p$-adic integers. Let $a_0$ be a $p$-adic
integer such that for some $i\geq0$ we have
$$\begin{array}{ll}
f(a_0)\equiv0(\operatorname{mod }p^{2i+1}),\\
f^\prime(a_0)\equiv0(\operatorname{mod }p^{i}),\quad f^\prime(a_0)\not\equiv0(\operatorname{mod }p^{i+1}).
\end{array}
$$
Then $f(x)$ has a unique $p$-adic integer root $x_0$ which satisfies
$x_0\equiv a_0(\operatorname{mod }p^{i+1})$.
\end{lemma}

For $a\in \mathbb Q_p$ and $r> 0$ we denote
$$B(a, r) = \{x\in \mathbb Q_p : |x-a|_p < r\}.$$

$p$-adic {\it logarithm} is defined by the series
$$\log_p(x) =\log_p(1 + (x-1)) =
\sum_{n=1}^\infty (-1)^{n+1}{(x-1)^n\over n},$$ which converges
for  $x\in B(1, 1)$;  $p$-adic {\it exponential} is defined by
$$\exp_p(x) =\sum^\infty_{n=0}{x^n\over n!},$$
which converges
for $x \in B(0, p^{-1/(p-1)})$.
\begin{lemma}\label{lem2} Let $x\in B(0, p^{-1/(p-1)})$. Then
$$|\exp_p(x)|_p = 1,\ \ |\exp_p(x)-1|_p = |x|_p, \ \
|\log_p(1 + x)|_p = |x|_p,$$
$$\log_p(\exp_p(x)) = x,\ \ \exp_p(\log_p(1 + x)) = 1 + x.$$
\end{lemma}
Denote
$$
\mathcal E_p=\left\{x\in\mathbb Q_p: |x|_p=1, |x-1|_p<p^{-1/(p-1)}\right\}.
$$

A more detailed description of $p$-adic calculus and $p$-adic
mathematical physics can be found in \cite{29,sc,48}.

Let $(X,{\mathcal B})$ be a measurable space, where ${\mathcal B}$
is an algebra of subsets of $X$. A function $\mu: {\mathcal B}\to
\mathbb Q_p$ is said to be a $p$-adic measure if for any $A_1, . . .
,A_n\in {\mathcal B}$ such that $A_i\cap A_j = \emptyset$, $i\ne
j$, the following holds:
$$\mu(\bigcup^n_{j=1}A_j)=\sum^n_{j=1}\mu(A_j).$$
A $p$-adic measure is called a probability measure if $\mu(X) =
1$. A $p$-adic probability measure $\mu$ is called {\it bounded}
if $\sup\{|\mu(A)|_p:A\in\mathcal B\}<\infty$ (see \cite{k91}).

We call a $p$-adic measure a probability measure \cite{16} if $\mu(X)=1$.

\subsection{\bf Cayley tree.}

The Cayley tree $\Gamma^k$
of order $ k\geq 1 $ is an infinite tree, i.e., a graph without
cycles, such that exactly $k+1$ edges originate from each vertex.
Let $\Gamma^k=(V, L)$ where $V$ is the set of vertices and  $L$
the set of edges.
Two vertices $x$ and $y$ are called {\it nearest neighbors} if
there exists an
edge $l \in L$ connecting them.
We shall use the notation $l=\langle x,y\rangle$.
A collection of nearest neighbor pairs $\langle x,x_1\rangle,
\langle x_1,x_2\rangle,...,\langle x_{d-1},y\rangle$ is called a
{\it
path} from $x$ to $y$. The distance $d(x,y)$ on the Cayley tree
is the number of edges of the shortest path from $x$ to $y$.

For a fixed $x^0\in V$, called the root, we set
\begin{equation*}
W_n=\{x\in V\,| \, d(x,x^0)=n\}, \qquad V_n=\bigcup_{m=0}^n W_m
\end{equation*}
and denote
$$
S(x)=\{y\in W_{n+1} :  d(x,y)=1 \}, \ \ x\in W_n, $$ the set  of
{\it direct successors} of $x$.

Let $G_k$ be a free product of $k + 1$ cyclic groups of the second
order with generators $a_1, a_2,\dots, a_{k+1}$,
respectively.
It is known that there exists a one-to-one correspondence between
the set of vertices $V$ of the Cayley tree $\Gamma^k$ and the group $G_k$.

\subsection{\bf $p$-adic SOS model}

Let $\mathbb Q_p$ be the field of $p$-adic numbers and
$\Phi$ be a finite set. A configuration $\sigma$ on $V$ is then defined
as a function $x\in V\mapsto\sigma(x)\in\Phi$; in a similar fashion
one defines a configuration $\sigma_n$ and $\sigma^{(n)}$ on $V_n$
and $W_n$ respectively. The set of all configurations on $V$
(resp. $V_n,\ W_n$) coincides with $\Omega=\Phi^V$
(resp.$\Omega_{V_n}=\Phi^{V_n},\ \Omega_{W_n}=\Phi^{W_n}$). Using
this, for given configurations $\sigma_{n-1}\in\Omega_{V_{n-1}}$
and $\omega\in\Omega_{W_n}$ we define their concatenations
by
$$
(\sigma_{n-1}\vee\omega)(x)=\left\{\begin{array}{ll}
\sigma_{n-1}(x),& \text{if}\  x\in V_{n-1},\\
\omega(x),& \text{if}\  x\in W_n.
 \end{array}\right.
$$
It is clear that $\sigma_{n-1}\vee\omega\in\Omega_{V_n}.$\\

Let $G^*_k$ be a subgroup of the group $G_k$. A function $h_x$ (for example,
a configuration $\sigma(x)$) of $x\in G_k$ is called $G^*_k$-periodic if
$h_{yx}=h_x$ (resp. $\sigma(yx)=\sigma(x)$) for any $x\in G_k$ and
$y\in G^*_k.$

A $G_k$-periodic function is called {\it translation-invariant}.

We consider {\it $p$-adic SOS model} on a Cayley tree,
where the spin takes values in the set
$\Phi:=\{1,2,\dots,m\}$, and is assigned to the vertices
of the tree.

The (formal) Hamiltonian is of a $p$-adic SOS form:
\begin{equation}\label{ham}
H(\sigma)=J\sum_{\langle x,y\rangle\in L}
\left|\sigma(x)-\sigma(y)\right|_\infty,
\end{equation}
where $J\in B(0, p^{-1/(p-1)})$ is a coupling constant,
$\langle x,y\rangle$ stands for nearest neighbor vertices
and $|\cdot|_\infty$ stands for usual absolute value.

The $p$-adic SOS model of this type can be considered as a generalization
of the $p$-adic Ising model (which arises when $m=1$).

\section{The system of $p$-adic vector-valued functional equations}
We use standard definition of a $p$-adic Gibbs measure, a translation-invariant
(TI) measure. Also, call measure $\mu$ {\it symmetric} if it is preserved
under simultaneous change $j\mapsto m-j$ at each vertex $x\in V$.

Let $z: x\mapsto z_x=(z_{0,x}, z_{1,x},\dots, z_{m,x})\in\mathcal E^{m+1}_p$
be a $p$-adic vector-valued function of $x\in V\setminus\{x_0\}$. Given
$n=1,2,\dots,$ consider the $p$-adic probability distribution $\mu^{(n)}$ on $\Omega_{V_n}$
defined by
\begin{equation}\label{mu_n}
\mu_{\tilde z}^{(n)}(\sigma_n)=Z_{n,\tilde z}^{-1}\exp_p\left\{H_n(\sigma_n)\right\}
\prod_{x\in W_n}{\tilde z}_{\sigma(x),x}.
\end{equation}
Here, $\sigma_n: x\in V_n\mapsto\sigma_n(x)$ and $Z_{n,\tilde z}$ is the corresponding
partition function:

\begin{equation}\label{Z_n}
Z_{n,\tilde z}=\sum_{\sigma\in\Omega_{V_n}}\exp_p\left\{H_n(\sigma_n)\right\}
\prod_{x\in W_n}{\tilde z}_{\sigma(x),x}.
\end{equation}

We say that the $p$-adic probability distributions (\ref{mu_n}) are
compatible if for all
$n\geq 1$ and $\sigma_{n-1}\in \Phi^{V_{n-1}}$:
\begin{equation}\label{cc}
\sum_{\omega_n\in\Omega_{W_n}}\mu_{\tilde z}^{(n)}(\sigma_{n-1}\vee \omega_n)
=\mu_{\tilde z}^{(n-1)}(\sigma_{n-1}).
\end{equation}
Here $\sigma_{n-1}\vee \omega_n$ is the concatenation of the configurations.\\
We note that an analog of the Kolmogorov extension theorem for distributions
can be proved for $p$-adic distributions given by (\ref{mu_n}) (see \cite{16}).
According to this theorem there exists a unique $p$-adic measure
$\mu_{\tilde z}$ on $\Omega$ such that,
for all $n$ and $\sigma_n\in \Omega_{V_n}$,
$$\mu_{\tilde z}(\{\sigma|_{V_n}=\sigma_n\})=\mu_{\tilde z}^{(n)}(\sigma_n).$$

Such a measure is called a {\it $p$-adic Gibbs measure} (pGM)
corresponding to the Hamiltonian (\ref{ham}) and vector-valued
function ${\tilde z}_x, x\in V$. We denote by $G(H)$ the set of all $p$-adic Gibbs
measures for the hamiltonian $H$. If $|G(H)|\geq2$ then we say that for this model
there is {\it a phase transition}.

The following statement describes conditions on ${\tilde z}_x$
guaranteeing compatibility of $\mu_{\tilde z}^{(n)}(\sigma_n)$.

\begin{pro}\label{pro1}
The $p$-adic probability distributions
$\mu^{(n)}_{\tilde z}(\sigma_n)$, $n=1,2,\dots$ in
(\ref{mu_n}) are compatible for $p$-adic SOS model iff for any $x\in V\setminus\{x^0\}$
the following system of equations holds:
\begin{equation}\label{p***}
 z_{i,x}=\prod_{y\in S(x)}\frac{\sum_{j=0}^{m-1}\theta^{|i-j|_\infty}z_{j,y}+\theta^{m-i}}
 {\sum_{j=0}^{m-1}\theta^{m-j}z_{j,y}+1}\ ,\quad i=0,1,\dots,m-1.
\end{equation}
Here $\theta=\exp_p(J)$ and $z_{i,x}=\frac{\tilde{z}_{i,x}}{\tilde{z}_{m,x}},\ i=0,1,\dots,m-1.$
\end{pro}
\begin{proof}
{\it Necessity}. Suppose that (\ref{cc}) holds; we want to prove (\ref{p***}).
Substituting (\ref{mu_n}) in (\ref{cc}), obtain, for any configurations
$\sigma_{n-1}\in\Omega_{V_n}$:
$$
\frac{Z_{n-1,\tilde z}}{Z_{n,\tilde z}}\sum_{\omega\in\Omega_{W_n}}\prod_{x\in W_{n-1}}
\prod_{y\in S(x)}\exp_p\left(J|\sigma_{n-1}(x)-\omega(y)|_\infty\right)\tilde{z}_{\omega(y),y}=
\prod_{x\in W_{n-1}}\tilde{z}_{\sigma_{n-1}(x),x}.
$$
From this for any $i\in\{0,1,\dots,m\}$ we get
\begin{equation}\label{p**}
\prod_{y\in S(x)}\frac{\sum_{j=0}^{m}\exp_p\left(J|i-j|_\infty\right)\tilde{z}_{j,y}}
 {\sum_{j=0}^{m}\exp_p\left(J(m-j)\right)\tilde{z}_{j,y}}=\frac{\tilde{z}_{i,x}}{\tilde{z}_{m,x}}.
\end{equation}
Denoting $\theta=\exp_p(J)$ and $z_{i,x}=\frac{\tilde{z}_{i,x}}{\tilde{z}_{m,x}}$, we get
(\ref{p***}) from (\ref{p**}).

{\it Sufficiency}. Suppose that (\ref{p***}) holds. It is equivalent to the
representations
\begin{equation}\label{p1*}
\prod_{y\in S(x)}\sum_{j=0}^m\theta^{|i-j|_\infty}\tilde{z}_{j,y}=a(x)\tilde{z}_{i,x},
\quad i=0,1,\dots,m-1
\end{equation}
for some function $a(x)$. We have

\begin{equation}\label{p1**}
\mbox{LHS of }(\ref{cc})=\frac{1}{Z_{n,\tilde{z}}}\exp_p\left(H(\sigma_{n-1})\right)
\prod_{x\in W_{n-1}}\prod_{y\in S(x)}\sum_{j=0}^m\theta^{|i-j|_\infty}\tilde{z}_{i,y}.
\end{equation}

Substituting (\ref{p1*}) into (\ref{p1**}) and denoting $A_n(x)=\prod_{y\in W_{n-1}}a(y)$,
we get
\begin{equation}\label{p1***}
\mbox{RHS of }(\ref{p1**})=\frac{A_{n-1}}{Z_{n,\tilde{z}}}\exp_p\left(H(\sigma_{n-1})\right)
\prod_{x\in W_{n-1}}\tilde{z}_{\sigma_{n-1}(x),x}.
\end{equation}

Since $\mu^{(n)}_{\tilde{z}},\ n\geq1$ is a probability, we should have
$$
\sum_{\sigma_{n-1}\in\Omega_{V_{n-1}}}\sum_{\omega\in\Omega_{W_n}}\mu^{(n)}_
{\tilde{z}}(\sigma_{n-1},\omega)=1.
$$

Hence from (\ref{p1***}) we get $Z_{n,\tilde{z}}=A_{n-1}Z_{n-1,\tilde{z}}$,
and (\ref{cc}) holds.
\end{proof}

The following proposition is straightforward.
\begin{pro}\label{pro2} Let $H$ be a hamiltonian of $m+1$-state $p$-adic
SOS model on a Cayley tree $\Gamma^k$. Then it hold the following
statements:

$1)$ Any measure $\mu$ with local distributions $\mu^{(n)}$ satisfying
(\ref{mu_n}),(\ref{cc}) belongs to $G(H)$.

$2)$ An $\mu\in G(H)$ is TI iff $z_{i,x}$ does not depend on $x: z_{i,x}\equiv z_i,\ x\in V,
\ i\in\Phi$,
and symmetric TI if and only if $z_i=z_{m-i},\ i\in\Phi$.
\end{pro}

\subsection{Boundedness of $p$-adic Gibbs measure}
\begin{thm}
Let $H$ be a $m+1$-state $p$-adic SOS model on a Cayley tree of order $k$.
Then a measure $\mu\in G(H)$ is bounded if and only if $m+1$ is not
divisible by $p$.
\end{thm}
\begin{proof}
Let $z_x=(z_{0,x},z_{1,x},\dots,z_{m-1,x})$ is a solution to (\ref{p***}).
Then from (\ref{p1*}) for all $x\in V\setminus\{x_0\}$
we find
$$
a(x)=\prod_{y\in S(x)}\left(\sum_{j=0}^{m-1}\theta^{m-j}z_{j,y}+1\right)=
\prod_{y\in S(x)}\left(\sum_{j=0}^{m-1}(\theta^{m-j}z_{j,y}-1)+m+1\right).
$$
Since $\theta\in\mathcal E_p$ and $z_x\in\mathcal E^{m}_p$, by
non-Archimedean norm's property we get
$$
|a(x)|_p=\left\{\begin{array}{ll}
1, & \mbox{if }\ p\nmid m+1;\\
\leq p^{-k}, & \mbox{if }\ p\mid m+1.
\end{array}\right.
$$
From this we obtain
\begin{equation}\label{A_n}
\left|A_{n}(x)\right|_p=\prod_{y\in W_{n-1}}|a(y)|_p=\left\{\begin{array}{ll}
1, & \mbox{if }\ p\nmid m+1;\\
\leq p^{-k|W_{n-1}|}, & \mbox{if }\ p\mid m+1.
\end{array}\right.
\end{equation}
We use the following recurrence formula
$$
Z_{n,z}=A_{n-1}Z_{n-1,z}.
$$
From (\ref{A_n}) we get
\begin{equation}\label{Z_nA}
\left|Z_{n,z}\right|_p=\prod_{x\in V_{n-1}}\left|A_{n-1}(x)\right|_p=
\left\{\begin{array}{ll}
1, & \mbox{if }\ p\nmid m+1;\\
\leq p^{-k|V_{n-1}|}, & \mbox{if }\ p\mid m+1.
\end{array}\right.
\end{equation}
For any configuration $\sigma\in\Omega_{V_n}$ by (\ref{Z_nA}) we have
$$
\left|\mu^{(n)}_z(\sigma)\right|_p=\frac{\left|\exp_p\{H_n(\sigma)\}\prod_{x\in W_n}
z_{\sigma(x),x}\right|_p}{\left|Z_{n,z}\right|_p}=$$
$$
\frac{1}{\left|Z_{n,z}\right|_p}=\left\{\begin{array}{ll}
1, & \mbox{if }\ p\nmid m+1;\\
\geq p^{k|V_{n-1}|}, & \mbox{if }\ p\mid m+1.
\end{array}\right.
$$
which means the measure $\mu_z$ is bounded if and only if $p\nmid m+1$.
\end{proof}

\section{Three-state SOS model}

From Proposition \ref{pro2} it follows that for any $z=\{z_x,\ x\in V\}$
satisfying (\ref{p***}) there exists a unique $p$-adic Gibbs measure $\mu$.
Denote by $\mbox{TIpGM}$ the set of all translation-invariant $p$-adic Gibbs measures
for the model (\ref{ham}). Note that $\mbox{TIpGM}\subset G(H)$. However,
description of the set $\mbox{TIpGM}$ for an arbitrary $m$ is not easy.
We now suppose that $m=2$, i.e.
$\Phi=\{0,1,2\}$. We assume that $z_{2,x}\equiv1$
($z_{m,x}\equiv1$ for general $m$).

\subsection{Translation-invariant solutions}

It is natural begin with TI solutions where $z_{x}=z\in\mathcal E^m_p$
is constant. Unless otherwise stated, we concentrate on the simplest case
where $m=2$, i.e., spin values are $0,1$ and $2$. In this case we obtain
from (\ref{p***}):
\begin{equation}\label{z_0}
z_0=\left(\frac{z_0+\theta z_1+\theta^2}{\theta^2 z_0+\theta z_1+1}\right)^k,
\end{equation}
\begin{equation}\label{z_1}
z_1=\left(\frac{\theta z_0+z_1+\theta}{\theta^2 z_0+\theta z_1+1}\right)^k.
\end{equation}

\begin{pro}\label{pro3}
If $p\neq3$ then the system of equations (\ref{z_0}),(\ref{z_1}) has no
solution in $\mathcal A_p=\{(z_0,z_1)\in\mathcal E^2_p: z_0\neq1\}$.
\end{pro}
\begin{proof}
Let $p\neq3$. Denote $A=z_0+\theta z_1+\theta^2$, $B=\theta^2 z_0+\theta z_1+1$.
Then from (\ref{z_0}) we get:
\begin{equation}\label{z_0-1}
(z_0-1)\left(B^k+(\theta^2-1)\sum_{i=0}^{k-1}A^{k-1-i}B^i\right)=0.
\end{equation}
Since $\theta,z_0,z_1\in\mathcal E_p$ we have
$$A\equiv3(\operatorname{mod }p)\quad\mbox{and}\quad B\equiv3(\operatorname{mod }p).$$
Consequently,
\begin{equation}\label{Bk}
\left|B^k+(\theta^2-1)\sum_{i=0}^{k-1}A^{k-1-i}B^i\right|_p=\left|3^k\right|_p=1.
\end{equation}
From (\ref{z_0-1}) and (\ref{Bk}) we obtain $z_0=1$.
\end{proof}
Observe that $z_0=1$ satisfies equation (\ref{z_0}) independently of
$k, \theta$ and $z_1$.
Substituting
$z_0=1$ into (\ref{z_1}), we have
\begin{equation}\label{z_1*}
z_1=\left(\frac{2\theta+z_1}{\theta^2+\theta z_1+1}\right)^k.
\end{equation}
We consider the function
$$
f(x)=\left(\frac{2a+x}{a^2+ax+1}\right)^k.
$$
Let $p\neq3$ and $a\in\mathcal E_p$. Then for any $x\in\mathcal E_p$ from non-Archimedean
norm's property we get
$$
\left|f(x)\right|_p=\left(\frac{\left|2a+x\right|_p}{\left|a^2+ax+1\right|_p}\right)^k=1,
$$
and
$$
\left|f(x)-1\right|_p=\left|\frac{(a-1)(1-x-a)\sum_{i=0}^{k-1}(2a+x)^{k-1-i}(a^2+ax+1)^i}{(a^2+ax+1)^k}\right|_p<p^{-1/(p-1)}.
$$
Thus, we have shown that $f:\mathcal E_p\mapsto \mathcal E_p$ if $p\neq3$ and $a\in\mathcal E_p$.

Now we shall show that $|f'(x)|_p<1$ for any $x\in\mathcal E_p$.
$$
f'(x)=-\frac{k(a^2-1)}{\left(a^2+ax+1\right)^2}\left(\frac{2a+x}{a^2+ax+1}\right)^{k-1}.
$$
Since $p\neq3$ and $a\in\mathcal E_p$ we obtain
$$
|f'(x)|_p=\left|k\left(a^2-1\right)\right|_p\leq\frac{1}{p}.
$$
Hence, we get
$$
\left|f(x)-f(y)\right|_p\leq\frac{1}{p}|x-y|_p\qquad\mbox{for all }x,y\in\mathcal E_p.
$$

Consequently, the function $f$ has a unique fixed point $x^*$ as $\mathcal E_p$ is compact.
Thus, we have proved the following proposition
\begin{pro}\label{pro4}
If $p\neq3$ then the system of equations (\ref{z_0}),(\ref{z_1}) has a unique
solution in $\mathcal E^2_p\setminus\mathcal A_p$.
\end{pro}
From Proposition \ref{pro3} and Proposition \ref{pro4} we get the following
\begin{thm}
Let $p\neq3$. Then there exists a unique translation-invariant $p$-adic
Gibbs measure for the three-state $p$-adic SOS model on a Cayley tree of order $k$.
\end{thm}

\begin{pro}\label{pro5} Let $p=3$. If
$\theta\in\{x\in\mathcal E_3: |x-1|_3<3^{-2}\}$ then the system of equations
(\ref{z_0}),(\ref{z_1}) has a solution in $\mathcal E_3^2$.
\end{pro}
\begin{proof} Let $p=3$. We consider the following polynomial with
$p$-adic integers coefficients:
\begin{equation}\label{g(x)}
g(x)=\theta x^{k+1}-x^k+(\theta^2+1)x-2\theta.
\end{equation}
We will check conditions of Hensel's lemma.
\begin{equation}\label{hensel}
g(1)=\theta(\theta-1)\quad\mbox{and}\quad g^\prime(1)=(\theta-1)(\theta+k+2)+3.
\end{equation}
Since $\theta\in\{x\in\mathcal E_3: |x-1|_3<3^{-2}\}$ we get
\begin{equation}\label{theta}
\theta\equiv1(\operatorname{mod }27).
\end{equation}
Using this, from (\ref{hensel}) we have
$$\begin{array}{ll}
g(1)\equiv0(\operatorname{mod }27),\\[2mm]
g^\prime(1)\equiv0(\operatorname{mod }3),\quad g^\prime(1)\not\equiv0(\operatorname{mod }9).
\end{array}
$$
Thus, we have shown that for the function $g(x)$ all conditions of Hensel's lemma are satisfied.
Then the function $g(x)$ has a unique $p$-adic integer root $x_*$ which satisfies
$x_*\equiv1(\operatorname{mod }9)$.

It is easy to see $x_*^k\in\mathcal E_3$ for any $k\geq1$ and $z_1^*=x_*^k$ is a
solution to (\ref{z_1*}). Consequently, $(1,z_1^*)\in\mathcal E_3^2$ is a
solution to (\ref{z_0}),(\ref{z_1}).
\end{proof}

\begin{pro}\label{pro6} Let $p=3$. If
$k$ is an even and it is not divisible by $p$ then a system of equations
(\ref{z_0}),(\ref{z_1}) has a solution in $\mathcal E_3^2$.
\end{pro}
\begin{proof} Let $k$ is an even and it is not divisible by $3$. Then we get
\begin{equation}\label{2^k}
2^k\equiv1(\operatorname{mod }3)\quad\mbox{and}\quad
k\not\equiv0(\operatorname{mod }3).
\end{equation}
Consider the function $g(x)$ (see (\ref{g(x)})). We have
$$
g(2)=2^{k+1}\theta-2^k+2(\theta^2-\theta+1)
$$
and
$$
g^\prime(2)=(k+1)2^k\theta-k2^{k-1}+\theta^2+1.
$$
Since $\theta\in\mathcal E_3$ and using (\ref{2^k}),
we obtain
$$
g(2)=2^{k+1}\theta-2^k+2(\theta^2-\theta+1)\equiv0(\operatorname{mod }3),
$$
and
$$
g^\prime(2)=(k+1)2^k\theta-k2^{k-1}+\theta^2+1\not\equiv0(\operatorname{mod }3).
$$
Then by Hensel's lemma there exists a unique $p$-adic integer $\tilde{x}_{*}$
such that $g(\tilde{x}_{*})=0$
and $\tilde{x}_{*}\equiv2(\operatorname{mod }3)$.
It is easy to check that $\tilde{x}_*^k\in\mathcal E_3$. Denote $\tilde{z}_*=\tilde{x}_*^k$.
Then $\tilde{z}_*$ is a solution to (\ref{z_1}). Consequently,
$(1,\tilde{z}_1^*)\in\mathcal E_3^2$ is a
solution to (\ref{z_0}),(\ref{z_1}).
\end{proof}

\begin{cor}\label{cor} Let $p=3$ and
$\theta\in\{x\in\mathcal E_3: |x-1|_3<3^{-2}\}$. If
$k$ is an even and it is not divisible by $p$ then a system of equations
(\ref{z_0}),(\ref{z_1}) has at least two solutions in $\mathcal E_3^2$.
\end{cor}
\begin{thm}
Let $p=3$ and $\theta\in\{x\in\mathcal E_3: |x-1|_3<3^{-2}\}$.
If
$k$ is an even and it is not divisible by $p$ then a phase transition occurs for the
three state $p$-adic SOS model on a Cayley tree of order $k$.
\end{thm}
\begin{pro}\label{pro7}
Let $p=3$ and $k=2$.\\
1) If
$\theta\not\in\{x\in\mathcal E_3: |x-10|_3<3^{-2}\}$ then the system of equations
(\ref{z_0}),(\ref{z_1}) has no solutions in
$\mathcal A_3=\left\{(z_0,z_1)\in\mathcal E_3^2: z_0\neq1\right\}$.\\
2) If
$\theta\in\{x\in\mathcal E_3: |x-37|_3<3^{-3}\}$ then the system of equations
(\ref{z_0}),(\ref{z_1}) has two solutions in
$\mathcal A_3=\left\{(z_0,z_1)\in\mathcal E_3^2: z_0\neq1\right\}$.
\end{pro}
\begin{proof} Let $p=3$ and $k=2$.
Since $z_i\in\mathcal E_p,\ i=1,2$ by Theorem \ref{tx2} there exist $\sqrt{z_i}$
in $\mathbb Q_p$. Denote
$x=\sqrt{z_0}$ and $y=\sqrt{z_1}$. Rewrite (\ref{z_0}) and (\ref{z_1}) as
\begin{equation}\label{x}
x=\frac{x^2+\theta y^2+\theta^2}{\theta^2 x^2+\theta y^2+1},
\end{equation}
\begin{equation}\label{y}
y=\frac{\theta x^2+y^2+\theta}{\theta^2 x^2+\theta y^2+1}.
\end{equation}
From (\ref{x}) we get
\begin{equation}\label{x-1}
(x-1)\left(\theta^2x^2+\theta y^2+1+(\theta^2-1)(x+1)\right)=0.
\end{equation}
Let $x\neq1$. Then from (\ref{x-1}) we have
\begin{equation}\label{x,y1}
\theta^2x^2+\theta y^2+1=(1-\theta^2)(x+1).
\end{equation}
Put this to (\ref{y}) and obtain
\begin{equation}\label{x,y2}
\theta y=\frac{x}{x+1}.
\end{equation}
Hence we conclude that
$$
x\equiv1(\operatorname{mod }3)\quad\mbox{and}\quad y\equiv2(\operatorname{mod }3).
$$
From (\ref{x,y1}) and (\ref{x,y2}) we get
\begin{equation}\label{x4}
\theta^3x^4+\theta(3\theta^2-1)x^3+(4\theta^3-2\theta+1)x^2+\theta(3\theta^2-1)x+\theta^3=0.
\end{equation}
Note that $x=1$ is a solution to (\ref{x4}) if and only if $12\theta^3-4\theta+1=0$.

Let $12\theta^3-4\theta+1=0$. Then from (\ref{x4}) we get
$$
(x-1)^2\left(\theta^3x^2+\theta\left(5\theta^2-1\right)x+\theta^3\right)=0.
$$
Consider the following equation
\begin{equation}\label{kv}
\theta^3x^2+\theta\left(5\theta^2-1\right)x+\theta^3=0.
\end{equation}
This equation is solvable in $\mathbb Q_3$ if and only if $\sqrt{1-7\theta^2}$
exists in
$\mathbb Q_3$. We will show that the quadratic equation (\ref{kv}) is not
solvable in $\mathbb Q_3$.
From $12\theta^3-4\theta+1=0$ and $\theta\in\mathcal E_3$ we get
$\theta\equiv10(\operatorname{mod }27)$. Then
we have
$$
1-7\theta^2\equiv3(\operatorname{mod }27).
$$
Hence by Theorem \ref{tx2} there does not exist $\sqrt{1-7\theta^2}$
in
$\mathbb Q_3$. Consequently, the equation (\ref{kv}) is not solvable
in $\mathbb Q_3$.

Let $12\theta^3-4\theta+1\neq0$. Denote $t=x+x^{-1}$. Then from
(\ref{x4}) we get
\begin{equation}\label{kv2}
\theta^3t^2+\theta\left(3\theta^2-1\right)t+2\theta^3-2\theta+1=0.
\end{equation}
Denote $D=\theta^4+2\theta^2-4\theta+1$. The equation (\ref{kv2})
has two distinct solutions if there exists $\sqrt{D}$ in $\mathbb Q_3$.
Moreover,
we have $t_i\equiv2(\operatorname{mod }3),\ i=1,2$.

The solutions of the equation (\ref{x4}) are
$$
x^\pm=\frac{1-3\theta^2+\sqrt{D}\pm\sqrt{\left(1-7\theta^2+\sqrt{D}\right)
\left(1+\theta^2+\sqrt{D}\right)}}{4\theta^2},
$$
$$
x_\pm=\frac{1-3\theta^2-\sqrt{D}\pm\sqrt{\left(1-7\theta^2-\sqrt{D}\right)
\left(1+\theta^2-\sqrt{D}\right)}}{4\theta^2}.
$$
Note that the existence of solutions $x^\pm$ equivalent to the existence
of $\sqrt{D}$ and the existence of
$\sqrt{2\left(1-7\theta^2-\sqrt{D}\right)}$ in $\mathbb Q_3$. Moreover,
if the solutions $x^\pm$ exists then
$x^\pm\equiv1(\operatorname{mod }3)$.

We will show that the number $\sqrt{2\left(1-7\theta^2+\sqrt{D}\right)}$
does not exist in $\mathbb Q_3$.
At first we will check the existence of $\sqrt{D}$. We have
$$
D=\theta^4+2\theta^2-4\theta+1=(\theta-1)\left(4\theta+(\theta-1)(\theta+1)^2\right).
$$
Since $\theta\in\mathcal E_3$ by Theorem \ref{tx2} we conclude that
the $\sqrt{D}$ exists if and only if
\begin{equation}\label{theta3}
\theta=1+3^{2n}(1+\varepsilon),\quad n\in\mathbb N,\ |\varepsilon|_3<1.
\end{equation}
Then we get
$$
\sqrt{D}=3^n(1+\varepsilon'),\qquad |\varepsilon'|_3<1.
$$
Hence, for all $n\in\mathbb N$ we have
$$
\left|1-7\theta^2+\sqrt{D}\right|_3=\left|1-7\left(1+3^{2n}
(1+\varepsilon)\right)+3^n(1+\varepsilon')\right|_3=
$$
\begin{equation}\label{+}
\left|-6+3^n+3^n\varepsilon'-14(1+\varepsilon)3^{2n}+7(1+\varepsilon)^2
3^{4n}\right|_3=\frac{1}{3}.
\end{equation}
By Theorem \ref{tx2} the number $\sqrt{2\left(1-7\theta^2+\sqrt{D}\right)}$
does not exist in $\mathbb Q_3$.

Now we will check the existence of $\sqrt{2\left(1-7\theta^2-\sqrt{D}\right)}$.
Using (\ref{theta3}) we obtain
$$
\left|1-7\theta^2+\sqrt{D}\right|_3=\left\{\begin{array}{cc}
\frac{1}{3},\quad\mbox{if }n>1;\\[2mm]
<\frac{1}{3},\quad\mbox{if }n=1.
\end{array}\right.
$$
Hence, by Theorem \ref{tx2} there does not exist
$\sqrt{2\left(1-7\theta^2-\sqrt{D}\right)}$ in $\mathbb Q_3$ if $n>1$.
So, we must check the case $n=1$.

Let $\left|\theta-37\right|_3<3^{-3}$, i.e., $\theta=37+\beta$, $|\beta|_3<3^{-3}$.
Then we obtain
$$
\left(1-7\theta^2+\sqrt{D}\right)\left(1-7\theta^2-\sqrt{D}\right)=
4\theta\left(12\theta^3-4\theta+1\right)=
$$
\begin{equation}\label{+-}
4(37+\beta)\left(12(37+\beta)^3-4(37+\beta)+1\right)=27(1+2\cdot3+3^2+\cdots)
\end{equation}
On the other hand by (\ref{+}) one can find
$$
1-7\theta^2+\sqrt{D}=6\left(1+\alpha_1\cdot3+\alpha_2\cdot3^2+\cdots\right),
\quad\alpha_i\in\{0,1,2\},\ i=1,2,\dots
$$
From this and by (\ref{+-}) we get
$$
2\left(1-7\theta^2-\sqrt{D}\right)=9\left(1+\gamma_1\cdot3+\gamma_2\cdot3^2+
\cdots\right),\quad\gamma_i\in\{0,1,2\},\ i=1,2,\dots
$$
Then by Theorem \ref{tx2} there exists $\sqrt{2\left(1-7\theta^2-\sqrt{D}\right)}$
in $\mathbb Q_3$.
Hence, $(x_+,y_+)$ and $(x_-,y_-)$ are the solutions to (\ref{x}),(\ref{y}),
where $y_\pm=\frac{x_\pm}{\theta(x_\pm+1)}$.
Redenote $z^\pm_0=x_\pm$ and $z^\pm_1=y_\pm$. It is clear that
$z^\pm_i\in\mathcal E_3,\ i=0,1$ and $z^\pm_0\neq1$.

Thus, we have shown that
the system of equations (\ref{z_0}),(\ref{z_1}) has two solutions in
$\mathcal A_3$ if $|\theta-37|_3<3^{-3}$ and
it has no solution in $\mathcal A_3$ if
$|\theta-1|_3<3^{-2}$.
\end{proof}

Thus by Proposition \ref{pro7} we have
\begin{thm} A phase
transition occurs for the three state $3$-adic
SOS model on a Cayley tree of order two if one of the following statements hold:\\
1) $\theta\in\{x\in\mathcal E_3: |x-1|_3<3^{-2}\}$;\\
2) $\theta\in\{x\in\mathcal E_3: |x-37|_3<3^{-3}\}$.
\end{thm}

\section{The uniqueness of $p$-adic Gibbs measures}

In the previous section we have shown that if $p\neq3$ then there is no phase transition
for the three state $p$-adic SOS model. A natural question arises:
what should be the relation between a number $m$ and prime $p$ in order to have
a phase transition for the $m+1$-state $p$-adic SOS model? In this section we
shall find this relation.

Let us first prove some technical results.

\begin{lemma}\label{lem3}\cite{MRasos}
If $a_i\in\mathcal E_p,\ i=1,2,\dots,n$ then $\prod_{i=1}^na_i\in\mathcal E_p$.
\end{lemma}
Recall that the $p$-adic norm of $x\in\mathbb Q^m_p$ defined as
$$
\Vert x\Vert_p=\max_{1\leq i\leq m}\{|x_i|_p\}.
$$

Let the collection of functions $F_i:\mathbb Q^{m+1}_p\mapsto\mathbb Q_p,\ i=0,1,\dots,m$
given by
$$
F_i(z,a_i,b_i,m)=\frac{\sum_{j=0}^{m}a_{ij}z_j}{\sum_{j=0}^{m}b_{ij}z_j},
$$
where
$$
a_{i}=(a_{i0},a_{i1},\dots,a_{im})\in\mathcal E^{m+1}_p\quad\mbox{and}\quad b_{i}=(b_{i0},b_{i1},\dots,b_{im})\in\mathcal E^{m+1}_p,
$$
for any $i=0,1,\dots,m$. For convenience we write $F_i(z)$ instead of $F_i(z,a_i,b_i,m)$.
\begin{lemma}\label{lem4}
Let $m+1$ is not divisible by $p$. Then $F_i$ is a function from
$\mathcal E^{m+1}_p$ to $\mathcal E_p$ for any $i=0,1,\dots,m$.
\end{lemma}
\begin{proof}
Let $z\in\mathcal E^{m+1}_p$. We will show that
$$
\left|F_i(z)-1\right|_p<p^{-1/(p-1)}.
$$

By Lemma \ref{lem3} we have $a_{ij}z_j,\ b_{ij}z_j\in\mathcal E_p$ for all $j=0,1,\dots,m$
as $a_{ij},\ b_{ij},\ z_j\in\mathcal E_p$. Since $m+1\not\equiv0(\operatorname{mod }p)$ and using
non-Archimedean norm's property we get
\begin{equation}\label{l21}
\left|\sum_{j=0}^mb_{ij}z_j\right|_p=\left|\sum_{j=0}^m\left(b_{ij}z_j-1\right)+m+1\right|_p=1.
\end{equation}
From this we obtain
$$
\left|F_i(z)-1\right|_p=\frac{\left|\sum_{j=0}^m(a_{ij}-b_{ij})z_j\right|_p}
{\left|\sum_{j=0}^mb_{ij}z_j\right|_p}\leq\max_{0\leq j\leq m}\{|a_{ij}-b_{ij}|_p\}<p^{-1/(p-1)}.
$$
\end{proof}
\begin{lemma}\label{lem5}
Let $m+1$ is not divisible by $p$. Then for any $z,t\in\mathcal E^{m+1}_p$ it holds
$$
\left|F_i(z)-F_i(t)\right|_p\leq\frac{1}{p}\Vert z-t\Vert_p,\qquad i=0,1,\dots,m.
$$
\end{lemma}
\begin{proof} Let $m+1\not\equiv0(\operatorname{mod }p)$. Then for any
$z,t\in\mathcal E_p$ from (\ref{l21}) we get
$$
\left|F_i(z)-F_i(t)\right|_p=\frac{\left|\sum_{j=0}^m\sum_{l=0}^ma_{ij}b_{il}(z_jt_l-z_lt_j)\right|_p}
{\left|\sum_{j=0}^ma_{ij}z_j\sum_{l=0}^mb_{ij}t_l\right|_p}=
$$
$$
\left|\sum_{j=0}^m\sum_{l=0}^ma_{ij}b_{il}\left(t_l(z_j-t_j)-t_j(z_l-t_l)\right)\right|_p=
$$
\begin{equation}\label{l31}
\left|\sum_{s=0}^m(z_s-t_s)\left(\sum_{l\neq s}a_{is}b_{il}-\sum_{j\neq s}a_{ij}b_{is}\right)\right|_p
\leq\frac{1}{p}\Vert z-t\Vert_p.
\end{equation}
As $a_{ij}, b_{il}\in\mathcal E_p$ then by Lemma \ref{lem3} we get $a_{ij}b_{il}\in\mathcal E_p$.
Hence, for any $i,j,l\in\{0,1,\dots,m\}$ we have
$$
\left|a_{ij}-b_{il}\right|_p=\left|a_{ij}-1+1-b_{il}\right|_p\leq
\max\{\left|a_{ij}-1\right|_p,\left|b_{il}-1\right|_p\}<p^{-1/(p-1)}\leq\frac{1}{p}.
$$
From this and by (\ref{l31}) we obtain
$$
\left|F_i(z)-F_i(t)\right|_p\leq\frac{1}{p}\Vert z-t\Vert_p.
$$
\end{proof}
\begin{lemma}\label{lem6}
Let $\{z^{(r)}\}_{r\in\mathbb N}$ and $\{t^{(r)}\}_{r\in\mathbb N}$ are the sequences
in $\mathcal E^{m+1}_p$. Then for any $n\in\mathbb N$ it holds
\begin{equation}\label{l41}
\left|\prod_{s=1}^nF_i(z^{(s)})-\prod_{s=1}^nF_i(t^{(s)})\right|_p\leq\frac{1}{p}
\max_{1\leq s\leq n}\left\Vert z^{(s)}-t^{(s)}\right\Vert_p,\qquad i=0,1,\dots,m.
\end{equation}
\end{lemma}
\begin{proof}
We prove this by induction on $n$. By Lemma \ref{lem5} condition (\ref{l41})
is satisfied for $n=1$. Suppose that (\ref{l41}) is valid for $n$.
We will prove that it is satisfied for $n+1$. We have
$$
\left|\prod_{s=1}^{n+1}F_i(z^{(s)})-\prod_{s=1}^{n+1}F_i(t^{(s)})\right|_p=
\left|\prod_{s=1}^{n}F_i(z^{(s)})F_i(z^{(n+1)})-\prod_{s=1}^{n}F_i(t^{(s)})F_i(t^{(n+1)})\right|_p=
$$
\begin{equation}\label{l42}
\left|\left(\prod_{s=1}^{n}F_i(z^{(s)})-\prod_{s=1}^{n}F_i(t^{(s)})\right)F_i(z^{(n+1)})+
\left(F_i(z^{(n+1)})-F_i(t^{(n+1)})\right)\prod_{s=1}^nF_i(t^{(s)})\right|_p
\end{equation}
By Lemma \ref{lem3} and Lemma \ref{lem4} we get
$$
\left|F_i(z^{(n+1)})\right|_p=1\qquad\mbox{and}\qquad\left|\prod_{s=1}^nF_i(t^{(s)})\right|_p=1.
$$
Using non-Archimedean norm's property from (\ref{l42}) we obtain
$$
\left|\prod_{s=1}^{n+1}F_i(z^{(s)})-\prod_{s=1}^{n+1}F_i(t^{(s)})\right|_p\leq
\max\left\{\left|\prod_{s=1}^{n}F_i(z^{(s)})-\prod_{s=1}^{n}F_i(t^{(s)})\right|_p,
\left|F_i(z^{(n+1)})-F_i(t^{(n+1)})\right|_p\right\}
$$
$$
\leq\frac{1}{p}\max\left\{\max_{1\leq s\leq n}\left\Vert z^{(s)}-t^{(s)}\right\Vert_p,
\left\Vert z^{(n+1)}-t^{(n+1)}\right\Vert_p\right\}=\frac{1}{p}\max_{1\leq s\leq n+1}\left\Vert z^{(s)}-t^{(s)}\right\Vert_p.
$$
\end{proof}
Consider the following functional equations
\begin{equation}\label{pro31}
z_{i,x}=\prod_{y\in S(x)}F_i(z),\quad i=0,1,\dots,m,\quad x\in V\setminus\{x_0\}.
\end{equation}
\begin{pro}\label{pro9} Let $m+1$ is not divisible by $p$. Then (\ref{pro31})
has a unique solution.
\end{pro}
\begin{proof}
Write
\begin{equation}\label{pro32}
\mathcal F_i(z_x)=\prod_{y\in S(x)}F_i(z_x).
\end{equation}
Let $z_x,t_x\in\mathcal E^{m+1}_p$ for any $x\in V\setminus\{x_0\}$.
Since $|S(x)|=k$ for all $x\in V\setminus\{x_0\}$ then by Lemma \ref{lem6} we have
\begin{equation}\label{pro33}
\left|\mathcal F_i(z_x)-\mathcal F_i(t_x)\right|_p\leq\frac{1}{p}\max_{y\in S(x)}
\Vert z_y-t_y\Vert_p\qquad\mbox{for all }\ x\in V\setminus\{x_0\}.
\end{equation}
Denote
$$
\mathcal F(z)=\left(\mathcal F_0(z),\mathcal F_1(z),\dots,\mathcal F_m(z)\right).
$$

By Lemma \ref{lem4} $\mathcal F$ is a function from $\mathcal E^{m+1}_p$ to $\mathcal E^{m+1}_p$.
From (\ref{pro33}) for any $x\in V\setminus\{x_0\}$ we get
$$
\Vert\mathcal F(z_x)-\mathcal F(t_x)\Vert_p\leq\frac{1}{p}\max_{y\in S(x)}\Vert z_y-t_y\Vert_p,
$$
which means that the function $\mathcal F$ is contractive. Then the function has a unique fixed
point in $\mathcal E^{m+1}_p$.
\end{proof}
\begin{thm}\label{uniq}
Let $H$ be a hamiltonian of $m+1$-state $p$-adic SOS model on a Cayley tree $\Gamma^k$. If
$p\nmid m+1$ then $|G(H)|=1$. Moreover, a measure $\mu\in G(H)$ is a
translation-invariant and symmetric.
\end{thm}
\begin{proof}
Let $F_m(z)\equiv1$ and
$$
F_i(z)=\frac{\sum_{j=0}^{m-1}\theta^{|i-j|_\infty}z_{j,y}+\theta^{m-i}}
{\sum_{j=0}^{m-1}\theta^{m-j}z_{j,y}+1}\ ,\quad i=0,1,\dots,m-1.
$$
Then by Proposition \ref{pro1} and Proposition \ref{pro9} there is a unique
$p$-adic Gibbs measure for the $m+1$-state $p$-adic SOS model on a Cayley tree
$\Gamma^k$ if $p\nmid m+1$.

Denote
$$
\mathcal I=\{(z_0,z_1,\dots,z_m): z_j\in\mathcal E_p\mbox{ and } z_j=z_{m-j},\ j=0,\dots,m\}.
$$
It is easy to see that for any $i=1,2,\dots,m+1$ it holds
$F_i(z)=F_{m-i}(z)$ if
$z\in\mathcal I$.
From the proof of the previous Proposition a unique solution of (\ref{pro31}) belongs
to $\mathcal I$. Then by Proposition \ref{pro2} a unique $p$-adic Gibbs measure
is a symmetric TI.
\end{proof}
{\bf Acknowledgments.} The author is grateful to Professor U. A. Rozikov and Professor F. M.
Mukhamedov for the useful advice and
valuable comments.

\end{document}